\documentclass[11pt,a4paper] {amsart}

\usepackage{amsmath,amscd,amssymb,amsthm,amstext,mathrsfs,dcolumn,mathdots,
booktabs,enumerate}

\usepackage{geometry}
\usepackage{ifpdf}
\ifpdf
\usepackage[all,line,cmtip]{xy}
\else
\usepackage[all,line,cmtip,dvips]{xy}
\fi

\input{xypic}

\theoremstyle{plain}
   \newtheorem{theorem}{Theorem}[section]
   \newtheorem{proposition}[theorem]{Proposition}
   \newtheorem{lemma}[theorem]{Lemma}
   \newtheorem{corollary}[theorem]{Corollary}
   \theoremstyle{definition}
   
   \newtheorem{definition}[theorem]{Definition}
   \newtheorem{example}[theorem]{Example}
   \theoremstyle{remark}
   
   \newtheorem{remark}[theorem]{Remark}

\newcommand{\Chains}{{\mathcal C}}
\newcommand{\Cycles}{{\mathbf K}}
\newcommand{\ZZ}{{\mathbb Z}}

\newcommand{\CP}{{\mathbb{CP}}}
\newcommand{\im}{{\rm{im}}}
\newcommand{\colim}[1]{{\operatorname{\underset\longrightarrow{lim}}_k}}

\newcommand{\SymmPower}[2]{{{\rm Sym}^{#1}(#2)}}
\newcommand{\ReducedSymm}[2]{{\overline{{\SymmPower {#1}{#2}}}}}

\newcommand{\SymmGroup}[1]{{\mathfrak{S}_{#1}}}
\newcommand{\FundGroup}{{\pi}}

\begin{document}
\title[curvature of vortex moduli spaces]
{On the curvature of vortex moduli spaces}
\author{Marcel B\"okstedt}
\address{Center for the Topology and Quantization of Moduli Spaces, Institut for Matematiske Fag, Aarhus Universitet, Ny Munkegade 118, 8000 \AA rhus C, Denmark}
\email{marcel@imf.au.dk}
\author{Nuno M. Rom\~ao}
\address{Center for the Topology and Quantization of Moduli Spaces, Institut for Matematiske Fag, Aarhus Universitet, Ny Munkegade 118, 8000 \AA rhus C, Denmark; {\upshape and}
Instytut Matematyki, Universytet Jagiello\'nski, ul.\ \L ojasiewicza 6, 30-348 Krak\'ow, Poland}
\email{nromao@imf.au.dk}
\date{\today}
\begin{abstract} We use algebraic topology to investigate local curvature properties of the moduli spaces of gauged vortices on a closed Riemann surface. After computing the homotopy type of the universal cover of the moduli spaces (which are symmetric products of the surface), we
prove that, for genus $g>1$, the holomorphic bisectional curvature of the vortex metrics cannot always be nonnegative in the multivortex case, and this property extends to all K\"ahler metrics on certain symmetric products. Our result rules out an established and natural conjecture on the geometry of the moduli spaces.

\end{abstract}
\maketitle

\section{Introduction}

Gauged vortices~\cite{JafTauVM, ManSut} are of interest as static, stable configurations in various classical gauge field theories with topological solitons. In their simplest
realisation, they are defined as solutions $({\rm d}_A, \phi)$ to the vortex equations
\begin{eqnarray}
&\bar\partial_{A} \phi  = 0, & \label{vort1} \\
&\ast F_{A} + \mu \circ \phi =0 & \label{vort2}
\end{eqnarray}
on a Hermitian line bundle $\mathcal{L} \rightarrow \Sigma$, where $\Sigma$ is a Riemann surface with a 
K\"ahler metric represented by the Hodge operator $\ast$. The ``Higgs field''
$\phi$ is a smooth section of $\mathcal{L}\rightarrow \Sigma$ and $\bar\partial_{A}$ is the $(0,1)$-part of a unitary connection ${\rm d}_A$
with curvature $F_A$. The second equation depends on the choice of a global moment map
\[
\mu: {\mathcal L} \rightarrow \mathfrak{u}(1)^{*}\cong {\rm i}\, \mathbb{R}
\]
for the Hamiltonian action of the structure group on the fibres (each equipped with the standard, $\rm U$(1)-invariant symplectic form),
which we fix to be, in a unitary trivialisation,
\begin{equation}\label{momentm}
w \mapsto \mu(w)=\frac{1}{2{\rm i}}(||w||^2-\tau),
\end{equation}
where $\tau \in \mathbb{R}$. Both vortex equations (\ref{vort1}) and (\ref{vort2}) are invariant under the group
of (unitary) gauge transformations 
$({\rm d}_A,\phi) \mapsto ({\rm d}_A + u^{-1}{\rm d}u,u\phi)$, 
$u \in {\rm Aut}_\Sigma(\mathcal{L})\cong C^{\infty}(\Sigma,{\rm U}(1))$, and one is usually interested in solutions up to this action.
For most of the discussion in this paper, we shall consider $\Sigma$ to be a connected,
closed surface of genus $g\ge 1$, and sometimes we will denote it by $\Sigma_g$ for emphasis.
Although our focus will be on line bundles,
the vortex equations above make sense in much broader contexts, and
this has been used to compute and generalise the Gromov--Witten
invariants~\cite{BerDasWen,CieGaiMunSal}.

The positive integer $k:=\frac{1}{2 \pi{\rm i}}\int_{\Sigma} F_A$ expresses the first Chern
class $c_1(L)\in H^2(\Sigma; \mathbb{Z}) \cong \mathbb{Z}$ in units of the fundamental class, and
it can be thought of as a quantised magnetic charge or vortex number. So $k$ is the
degree of the effective divisor associated to the holomorphic section $\phi$,
whose zeroes (counting multiplicities) correspond to locations of individual vortex cores on the surface.
It follows directly from (\ref{vort2}) that a metric on $\Sigma$ can only
support $k$ vortices provided $\tau {\rm Vol}(\Sigma) \ge 4 \pi k$, independently of the choice of Hermitian structure on the line bundle. (We use ${\rm Vol}(\Sigma)$ to denote the total area of the surface.)
Conversely, Bradlow~\cite{BraVHLB} and Garc\'\i a-Prada~\cite{GarVRS} showed that, if the strict inequality 
\begin{equation} \label{Brad}
\tau {\rm Vol}(\Sigma) > 4 \pi k
\end{equation}
is satisfied,
one can construct a vortex solution from any effective divisor of degree $k$.
Thus the moduli space $\mathcal{M}_{k,g}$ of $k$-vortex sulutions on $\Sigma=\Sigma_g$ up to gauge transformations
is the symmetric product
\begin{equation}\label{modspaces}
\mathcal{M}_{k,g} \cong {\rm Sym}^k(\Sigma_g):=(\Sigma_g)^k / \mathfrak{S}_k.
\end{equation}
(Here $\mathfrak{S}_k$ denotes the symmetric group on $k$ letters, acting on the Cartesian product
$(\Sigma_g)^k$ by permuting the factors.)
All these spaces are smooth manifolds,  equipped with a natural complex structure induced from the one
on $\Sigma_g$. We will always assume the parameter $\tau$
to be fixed and to satisfy $\tau > \frac{4 \pi k}{{\rm Vol}(\Sigma)}$, whenever we refer to a moduli space of vortices.

In gauge theory, the result (\ref{modspaces}) 
already has interesting consequences at the purely topological level
(see~\cite{DonCMT} for an application); but it turns out that the moduli spaces of vortices have
a richer structure, as they also support natural Riemannian metrics.
On the tangent space to a vortex represented by $({\rm d}_A,\phi)$, one can define this metric pointwise as the restriction of the $L^2$-metric to solutions $(\dot A,\dot \phi)$ of the linearised vortex equations about $({\rm d}_A,\phi)$ orthogonal to the orbit of the gauge group. One can show that the natural complex structure is an isometry of this metric, and the two define a K\"ahler structure~\cite{SamVS}. Alternatively, one can
understand these K\"ahler structures in terms of the symplectic (or K\"ahler) reduction of
the $L^2$-metric on the affine space of all fields, with respect to the Hamiltonian action
of the group of gauge transformations~\cite{GarVRS}.

The geometry of the vortex metrics has
been a focus of interest, since it has a whole range of consequences to the
physical behaviour of vortices, which is very hard to study otherwise. For instance,
the geodesic flow associated to these metrics gives a good approximation
to slow scattering in the abelian Higgs model~\cite{StuAHM}, and the corresponding
symplectic form is relevant to construct other dynamical systems with potentially interesting
applications~\cite{ManFOVD, ColTonDCSV}. Natural extensions of these basic mechanical models might include fermions, or extra potential terms that are covariant with respect to the moduli space metrics, in order to model particular phenomenology. Such extra terms will typically be nonlinear and involve the curvature of the Levi-Civita connection associated to the metric on the moduli space~\cite{MosShiQMMS}. More generally, curvature properties have implications for spectral geometry, which becomes relevant when one goes on to discuss supersymmetric quantum mechanics, or to construct improved approximations to the dynamics in the underlying classical field theories.

Computing the metrics on the moduli spaces is still a very difficult task. Samols~\cite{SamVS} was
able to write down a local formula in terms of the moduli, using a remarkable localisation argument in the spirit of previous work by Strachan~\cite{StrLVS}. From a version of Samols' formula for closed surfaces, 
one has been able to extract some global information --- for example, the symplectic volume~\cite{ManNasVVMS} of
the moduli spaces, their total scalar curvature~\cite{BapL2M}, and other integrals~\cite{RomGVB} which
provide information about vortex interactions~\cite{EasRomCFHM}. All these quantities are determined by
the K\"ahler class; in contrast, qualitative information
about the local geometry, such as more detailed curvature properties, has proven more difficult to obtain.
In~\cite{ManSpeAIV}, Manton and Speight described two alternative schemes of studying vortex metrics in the regime of large separation among $k \ge 2$ vortex
cores in the euclidean plane $\Sigma = \mathbb{C}$. These yield an asymptotic
K\"ahler metric whose coefficients can be written in terms of Bessel functions,
and from it one calculates the corresponding Ricci tensor.  This can be seen to be positive semi-definite (at every point), with two-dimensional kernel consisting of the directions associated to rigid translations, 
and it leads to a scalar curvature $Scal$ which is strictly positive everywhere.

In this paper, we address curvature properties of the moduli
spaces that are consistent with the topological consequences of the result~(\ref{modspaces}), and which are in a sense stronger than the sign of the Ricci curvature. The sectional curvature tensor carries all the curvature information
about a Riemannian manifold, but in complex geometry it is natural to restrict it to
tangent 2-planes that are invariant under the complex structure $J$, and this is
what is called the holomorphic sectional curvature. Particularly suitable for 
considerations relating topology to curvature in K\"ahler geometry is the holomorphic bisectional curvature introduced by Goldberg and Kobayashi in~\cite{GolKobHBC},
which carries more information than the holomorphic sectional curvature but less information than the sectional curvature. 
Following the conventions in \cite{MokMRT}, it is defined for $J$-invariant 2-planes $\sigma, \sigma'$ in the tangent space at a point by
\begin{equation} \label{biseccurv}
B(\sigma,\sigma'):=\frac{R(X,JX,X',JX')}{\langle X,X\rangle \langle X',X' \rangle },
\end{equation}
where $X,X'$ are nonzero real vectors contained in $\sigma,\sigma'$, respectively, and $\langle \cdot,\cdot \rangle$ denotes the K\"ahler metric. (Clearly, this definition is independent of the real vectors chosen.) With respect to a basis of normal holomorphic vector fields containing vectors $\xi, \xi'$ along $\sigma, \sigma'$, (\ref{biseccurv}) can also be described as the component $R_{\xi \bar\xi\xi' \bar\xi'}$ of the Riemann tensor in holomorphic coordinates.
It is easy to see~\cite{GolKobHBC} that positivity (or negativity) of the holomorphic bisectional curvature
implies positivity (or negativity) of both the Ricci tensor and the holomorphic sectional curvature, the
latter being simply given by $H(\sigma)=B(\sigma,\sigma)$.

We say that a K\"ahler metric has nonnegative holomorphic bisectional
curvature if $B$ is positive or zero at every point, for all choices of $J$-invariant 2-planes.
Our main result is the following:

\begin{theorem} \label{THMholbiscur}
If $g > 1$ and $1 \leq k\le2g-2$, no K\"ahler metric on ${\rm Sym}^k (\Sigma_g)$ can have 
nonnegative holomorphic bisectional curvature.
\end{theorem}

\noindent
Of course, this result applies to the  $L^2$-metric of vortices
described above and is independent of the choice of metric and complex
structure on $\Sigma$, as well as of the Hermitian inner product on $\mathcal{L}$.
Theorem~\ref{THMholbiscur} is somehow unexpected in the light of what is
known about the curvature behaviour of 
related moduli spaces with their natural K\"ahler structures. For example, the moduli spaces
of Higgs bundles
over a compact Riemann surface carry natural metrics~\cite{HitSDERS}. Fibre-integral type formulas for their curvature
tensors have been obtained using deformation theory in terms of representatives
of Kodaira--Spencer classes, from which one can deduce
that these metrics do have nonnegative holomorphic sectional 
curvature~\cite{BisSchGMSHB}.
Similar positivity results have been
established for the curvature of natural metrics on moduli spaces of unitary bundles and 
holomorphic vector bundles over a compact Riemann surface~\cite{BisSchDBQM,
BisSchKSMSPB}.
All these results build up on earlier work of Itoh on the geometry of spaces of anti-selfdual connections~\cite{ItoGASDC}; 
see also~\cite{TakZogGMSVB} and~\cite{JosPenGMSSVB}.

An immediate application of Theorem~\ref{THMholbiscur} is that it rules out a natural conjecture of Biswas and Schumacher stated in reference~\cite{BisSchCVE} as Corollary~8.2 (without proof);
the authors assert that a related metric on the moduli space of stable triples $(\mathcal{E}_1, \mathcal{E}_2, \varphi)$ associated to solutions of coupled vortex equations~\cite{GarcPraDRSB, BradGarcPra}
should have nonnegative holomorphic bisectional curvature.
Here, ${\mathcal{E}_i}$ ($i=1,2$) are holomorphic vector bundles over a closed Riemann surface $\Sigma$ and $\varphi:\mathcal{E}_2 \rightarrow \mathcal{E}_2$ an
$\mathcal{O}_{\Sigma}$-linear holomorphic map. 
If one takes $\mathcal{E}_2$ to be the trivial vector bundle, $\varphi$ can be
identified with a section  $\phi: \Sigma \rightarrow \mathcal{E}_1$
(which determines the morphism $\varphi$ by fibrewise multiplication) satisfying (\ref{vort1}), and the coupled vortex equations reduce to the vortex equation (\ref{vort2}). Thus gauged vortices are recovered as a special case, and our result above provides counter-examples of this conjecture in the nontrivial multivortex case. In fact, in the case $k=1$, 
for which $\mathcal{M}_1$ (as a complex manifold) is $\Sigma$ itself, and assuming $g\ge 2$, this assertion also follows trivially from the theorem of Gau\ss--Bonnet, since the Gau\ss\ curvature and the holomorphic bisectional curvature
on Riemann surface can easily be seen to coincide. For this reason, our interest will be 
strictly in the situation where $k>1$.

The strategy of proof of Theorem~\ref{THMholbiscur} in our paper is the following. 
If ${\rm Sym}^{k}(\Sigma_{g})$ had a K\"a{}hler metric with 
nonnegative holomorphic bisectional curvature, so would its universal
cover $\widetilde{{\rm Sym}^{k}(\Sigma_{g})}$. Using classification results (specifically, ``rigidity'' results obtained
by Mok in the 1980s, cf.~\cite{MokMRT}), 
this would imply that all homology groups of the universal covering
space would be finitely generated. But we can compute these
groups using algebraic topology, and it turns out that, if $1<k\le 2g-2$,
the middle homology groups $H_{k}(\widetilde{{\rm Sym}^{k}(\Sigma_{g})};\mathbb{Z})$ are not finitely generated.

It was pointed out to us that an alternative proof of Theorem~\ref{THMholbiscur}
can be presented, based on results obtained by Fang~\cite{Fan} on
the characterisation of compact K\"ahler manifolds whose holomorphic
bisectional curvature satisfies a condition that is weaker than
the notion of nonnegativity used in the present paper. This argument
explores the fact that the Albanese variety~\cite{GriHar} of ${\rm Sym}^{k}(\Sigma)$
coincides with the Jacobian variety of $\Sigma$. General properties
of Abelian varieties then imply that, if ${\rm Sym}^{k}(\Sigma)$ had nonnegative
holomorphic bisectional curvature, it would have to be
a holomorphic fibration over this Jacobian, as a consequence of
Fang's Theorem C in~\cite{Fan}; but this leads to a contradiction for the
range of values of $k$ considered in the hypothesis of Theorem~\ref{THMholbiscur}.
Since the argument of proof sketched in the preceding paragraph is
more constructive, in the sense that it produces {new information} on
the topology of vortex moduli spaces as by-product, we believe
that our preference for the strategy of proof using Mok's rigidity
results, as given in the present paper, is not unjustified. We are also aware of~\cite{BisCurv}.

\begin{remark}
There is another point of view on the calculations in this paper. Our main strategy to prove  Theorem~\ref{THMholbiscur} is to
study the topology of the {\em universal cover} of the moduli space of vortices --- and this cover has independent interest,
since it plays an important role at the level of quantisation. To be more precise: the structure of supersymmetric quantum field theories is
 reflected in the topology and the Hodge theory of  the {universal covers} of their classical 
moduli spaces. One famous instance of this is provided by
the beautiful Sen conjectures~\cite{Sen}. 
There is a purely mathematical formulation of these conjectures, which amounts to precise statements about  $L^2$-cohomology of the universal covers
of the moduli spaces of centred magnetic monopoles in $\mathbb{R}^3$. 
The dimension of $L^2$-cohomology groups gives information about the spectrum of bound states at the lowest energy level,  
which can therefore be extracted directly from the underlying classical field theory. A variety of methods has been used to prove part of these assertions~\cite{Seg,SegSel,SetSteZas,HitL2}. 
There is a striking analogy between this body of results and some of the topological underpinnings of our paper. In particular, from Theorem~\ref{th:main} below one can establish vanishing results for certain $L^2$-Betti numbers 
and then, using standard techniques~\cite{Lue}, compute the nontrivial $L^2$-Betti numbers of the moduli spaces we are considering
is this paper. In ongoing work, we study the $L^2$-cohomology for vortices in the same spirit as the papers \cite{Seg,SegSel} mentioned above address  the Sen conjectures for
monopoles. 
\end{remark}

Let us present the layout of this paper. We start by sketching standard topological results on the symmetric
products of a closed oriented surface in Section~\ref{topSym}. 
There is a cell decomposition of ${\rm Sym}^{k}(\Sigma)$ due to
Kallel and Salvatore; in Section \ref{sec:CW} 
we recall this construction, and show how to lift it to obtain a cell
decomposition of the universal cover, for which we can write down the
cellular chain complex. In Section \ref{sec:Homology}, we do the
algebraic computation that calculates the homology groups of this chain complex.
For completeness, we determine the homotopy type of the universal
cover in Section \ref{sec:Homotopy}. In Section \ref{sec:Rigidity}
we conclude the proof of Theorem \ref{THMholbiscur}, and make some
final remarks about vortices in Section~\ref{hyperel}.

\begin{remark}
After this work was completed, Indranil Biswas wrote a preprint~\cite{Bisnew} proving that the hypothesis $k\le 2g-2$ in our Theorem~\ref{THMholbiscur} is in fact unnecessary,
using arguments
in algebraic geometry which apply strictly to the case $k > 2g-2$.
\end{remark}

\section{Symmetric products of a closed oriented surface}\label{topSym}

For a smooth closed oriented surface $\Sigma$ and a positive integer $k$, the $k$-fold symmetric product
\[
{\rm Sym}^{k}(\Sigma):= \Sigma^k/\mathfrak{S}_k
\]
is a smooth, compact $2k$-dimensional manifold. If $\Sigma$ is a complex curve (i.e.\ is given a complex structure), then ${\rm Sym}^{k}(\Sigma)$ inherits a complex structure, and it parametrizes degree-$k$ effective divisors on $\Sigma$. This justifies the importance of these spaces in classical algebraic geometry, and also why they appear as moduli spaces of $k$-vortices on $\Sigma$. 
Our main goal in this paper is to prove a curvature property for  K\"ahler metrics on these manifolds which is implied by the topology of their universal covers. To begin with, we shall recall some standard facts about the topology of the symmetric products themselves,
and make some preliminary observations that situate our main topological results within the algebraic-geometric context.

The cohomology ring of ${\rm Sym}^{k}(\Sigma)$ was computed by Macdonald in~\cite{MacSPAC}, and we will quote the result here for completeness. 
Suppose that $\Sigma\equiv \Sigma_g$ has genus $g\ge 1$.
Start by fixing a symplectic basis $\{ x_1,\ldots x_g, y_1,\ldots, y_g \}$ for
$H_1(\Sigma_g;\mathbb{Z})\cong\mathbb{Z}^{2g}$, i.e.\ 1-cycles for which
the relations
\begin{equation} \label{symplbasis}
\sharp \langle x_i, y_j \rangle = \delta_{ij}, \qquad\sharp \langle x_i,x_j \rangle =\sharp \langle y_i, y_j \rangle=0
\end{equation}
for the intersection pairing are satisfied. (This basis is unique up to a base change in ${\rm Sp}_{2g}\mathbb{Z}$.) The Poincar\'e-dual basis
$\{ \alpha_i\}_{i=1}^{2g}$ of $H^1(\Sigma_g;\mathbb{Z})$ then generates the integral cohomology ring of $\Sigma_g$, where we use labels (mod $2g$) such that all cup products
are determined by the relations
\begin{eqnarray*}
& \alpha_i \alpha_j =0, & i \ne j \pm g, \\
& \alpha_i \alpha_{i+g} = - \alpha_{i+g} \alpha_i =\beta, & 1 \le i \le g,
\end{eqnarray*}
where $\beta \in H^{2}(\Sigma_g;\mathbb{Z})$ is the fundamental class (see~\cite{FulAT}). Macdonald shows~\cite{MacSPAC} that $H^{\ast}({\rm Sym}^k(\Sigma_g);\mathbb{Z})\cong H^{\ast}((\Sigma_g)^k;\mathbb{Z})^{\mathfrak{S}_k}$, and this is
described as follows. For each $1 \le \ell \le g$, let $\pi_\ell:(\Sigma_g)^k \rightarrow \Sigma_g$ denote the projection onto the $\ell$th factor and consider
the classes
\[
\alpha_{i\ell}:=\pi_\ell^{*}\alpha_i, \quad \beta_\ell := \pi_\ell^{*}\beta
\quad \in
H^{\ast}((\Sigma_g)^k;\mathbb{Z})\cong H^{\ast}(\Sigma_g;\mathbb{Z})^{\otimes k},
\]
which satisfy the following relations in degree two:
\begin{eqnarray*}
& \alpha_{i\ell} \alpha_{j\ell} =0 & i \ne j \pm g, \\
& \alpha_{i\ell} \alpha_{i+g,\ell}= -\alpha_{i+g,\ell}\alpha_{i\ell}=\beta_\ell, & 1 \le i \le g, \\
&\alpha_{i \ell},\alpha_{j,\ell'}=-\alpha_{j\ell'}\alpha_{i \ell}, &\ell\ne \ell'
\end{eqnarray*}
We sum over the label $\ell$ and obtain the classes
\begin{eqnarray*}
& \xi_i := \sum_{\ell =1}^{k} \alpha_{i\ell} ,
& 1\le i \le 2g, \\
& \eta := \sum_{\ell=1}^{k}\beta_{\ell} 
\end{eqnarray*}
which descend to $H^{\ast}({\rm Sym}^k(\Sigma_g);\mathbb{Z})$ since they
are $\mathfrak{S}_k$-invariant.
For convenience, we also define
\[
\sigma_i := \xi_i \xi_{i+g}, \qquad 1\le i \le g.
\]
Then we have~\cite{MacSPAC,BerThaQCSP}:

\begin{theorem} \label{THMcohomSym}
The integral cohomology ring $H^{\ast}({\rm Sym}^{k}(\Sigma);\mathbb{Z})$ is
generated by the degree-2 class $\eta$ and the degree-1 classes
$\xi_1, \ldots \xi_{2g}$ (anticommuting with each other and commuting with $\eta$).
For each choice of disjoint subsets $I_1,I_2,J \subset N_g:=\{1,\ldots, g\}$, there is an
extra relation
\begin{equation}\label{relations}
\eta^r\prod_{i_1\in I_1}\xi_{i_1} \prod_{i_2 \in I_2}\xi_{i_2+g}\prod_{j\in J}(\eta-\sigma_j)=0,
\end{equation}
for each $r\in \mathbb{N}_0$ satisfying
\begin{equation}\label{r}
r \le k - |I_1| - |I_2| - 2 |J| + 1.
\end{equation}
If $k>2(g-2)$, all relations (\ref{relations}) are generated by the one with $J=N_g$ and $r=k-2g+1$. If $k \ge 2(g-1)$, they are generated by those for which $r\in \{0,1\}$ and equality holds in (\ref{r}).
\end{theorem}

Of course, the homology of ${\rm Sym}^k(\Sigma)$ can be read off from this result by Poincar\'e duality, since
these spaces are orientable as complex manifolds.  A more refined approach to this calculation was presented by Kallel and Salvatore in~\cite{KallelSalvatore}, where
a certain homotopy equivalence $\SymmPower k\Sigma \to \ReducedSymm k\Sigma$ is introduced via a quotienting procedure, together with an 
explicit cell decomposition of $\ReducedSymm k\Sigma$; then the cell homology of the quotient space can be used to compute the homology
of the symmetric products themselves. The Kallel--Salvatore construction will play an important r\^ole in our discussion, and we shall review it in Section~\ref{sec:CW}.

To understand the structure
of the universal covers  $\widetilde{{\rm Sym}^k(\Sigma)}$, one might start by fixing a complex structure on $\Sigma$, together with a basepoint $p \in \Sigma$ and a basis of holomorphic 1-forms. These data define an Abel--Jacobi map~\cite{GriHar} for each $k\ge 1$,
\begin{equation}\label{AJ}
{\rm Sym}^k(\Sigma) \longrightarrow  {\rm Jac} (\Sigma) = H^0(\Sigma, K_{\Sigma})^{*}/H_{1}(\Sigma;\mathbb{Z}),
\end{equation}
which can be interpreted as mapping divisors of degree $k$ on $\Sigma$ to equivalence classes of 
holomorphic line bundles. Sections of these bundles (modulo multiplication by scalars) vanish 
at the $k$ points the divisors specify, counted with multiplicities.
If $k>2g-2= \deg  K_{\Sigma}$, then there are no special divisors of degree $k$ (i.e. all divisors correspond to invertible sheaves with vanishing first cohomology group). Riemann--Roch then implies that the corresponding spaces of global sections have dimension $k-g+1$.
Thus we can interpret the pre-image of each line bundle via the map (\ref{AJ})  as a space
of sections modulo scalar multiplication; so the Abel--Jacobi map is a fibre bundle with
typical fibre $\mathbb{CP}^{k-g}$ in this situation.
The Jacobian variety ${\rm Jac}(\Sigma)$ is a real $2g$-torus, and one can draw the diagram
\[
\xymatrix{
&\widetilde{{\rm Sym}^k(\Sigma)}\ar@{.>}[r]\ar[d]& \mathbb{R}^{2g}\ar[d]\\
\mathbb{CP}^{k-g}\ar[r]\ar@{.>}[ru]&{\rm Sym}^{k}(\Sigma)\ar[r]&\mathbb{T}^{2g}
}
\]
(where the dashed maps exist by the elementary lifting criterion of covering maps~\cite{Hatcher}). Thus we conclude that for $k>2g-2$ the universal cover $\widetilde{{\rm Sym}^k(\Sigma)}$
is itself a $\mathbb{CP}^{k-g}$-bundle over a contractible space,  and it has therefore the same homotopy type as the fibres. (Notice that ${\rm Sym}^k{\Sigma_0}\cong\mathbb{CP}^k$, which is simply connected, so this result is trivially true for $g=0$.)

However, when $k \le 2g-1$ this analysis no longer applies and the topology is more complicated,
as the dimension of the fibres of (\ref{AJ}) will jump on special divisors. We shall establish that, in this case, the universal cover $\widetilde{{\rm Sym}^k(\Sigma)}$ turns out to have  the same homotopy type as the wedge sum of a complex projective space (whose dimension may be smaller than $k-g$) and an infinite bouquet of $k$-spheres --- cf.~Theorem~\ref{homotopytype}.

A basic fact that we shall use below is the following result:

\begin{proposition} \label{THMpi1Sym}
For all $k>1$, $\pi := \pi_1({\rm Sym}^k (\Sigma_g)) 
\cong H_1(\Sigma_g; \mathbb{Z}) \cong \mathbb{Z}^{2g}$.
\end{proposition}

\noindent
A proof by M.\ Roth was presented in~\cite{BerThaQCSP}. This involves looking at the map
$\Sigma \rightarrow {\rm Sym}^{k}(\Sigma)$ given by $z \mapsto z+(k-1)p$ (for a fixed basepoint $p \in \Sigma$) and showing that
it induces a surjection of fundamental groups, with Abelian image and kernel $[\pi_1(\Sigma),\pi_1(\Sigma)]$.

In the rest of this paper, we will always use $x_1, \ldots, x_g,y_1, \ldots, y_g$ to denote generators of the fundamental group $\pi$ of the symmetric product ${\rm Sym}^k (\Sigma_g)$. These generators can be regarded as images of a standard basis for $H_1(\Sigma_g;\mathbb{Z})$ as in (\ref{symplbasis}) under the isomorphism in Proposition~\ref{THMpi1Sym}.

\section{A cell decomposition of the universal cover}
\label{sec:CW}

In \cite{KallelSalvatore}, Kallel and Salvatore define a quotient map
$r:\SymmPower k\Sigma \to \ReducedSymm k\Sigma$, where $\Sigma=\Sigma_g$
is a closed oriented surface of genus $g$ . They prove that this 
quotient map is a homotopy equivalence, and give a cell
decomposition of $ \ReducedSymm k\Sigma$. They use this
cell decomposition to compute the homology of 
$\ReducedSymm k\Sigma$, and thus of the homology of $\SymmPower k\Sigma$.

Our main goal in this section is to explain how this cell decomposition can be lifted to study the
universal covers $\widetilde{{\rm Sym}^k(\Sigma)}$, and to describe the associated cell complex. 
We begin by recalling the original construction in some detail; there is nothing new in our 
presentation, except maybe that we make an effort to be gentle to
non-topologists. After lifting this construction, we will prove the following theorem:

\begin{theorem}
\label{th:chain_complex}
Let $\Sigma$ be a closed Riemann surface of genus $g$.
Let ${\rm Sym}^\infty(\Sigma)$ be the space defined by the direct limit $\colim{k}\widetilde{{\rm Sym}^k(\Sigma)}$.
There is a cell decomposition of ${\rm Sym}^\infty(\Sigma)$ such that the
corresponding cellular chain complex ${\mathcal C}_{*}({\rm Sym}^\infty(\Sigma))$ is    
the differential graded algebra 
\[
\ZZ[x_{1},\dots x_{g},y_{1},\dots
y_{g}]\otimes \Lambda[e_{1},\dots,e_{g},f_{1},\dots f_{g}]\otimes
\Gamma[\gamma_{2}].
\]
The degrees of the generators $x_{i},y_{j}$ are $0$, the degrees of the
generators
$e_{i},f_{j}$ are $1$, and the degree of $\gamma_{s}$ is $s$.
The boundary map is given by
\begin{align*}
\partial(e_{i})&=1-x_{i},\\
\partial(f_{i})&=1-y_{i},\\
\partial(\gamma_{s})&=\sum_{i}((1-y_{i})e_{i}-(1-x_{i})f_{i})\gamma_{s-1}.
\end{align*}
This cell decomposition restricts to cell decompositions of 
each $k$-skeleton ${\rm Sk}^{k}({\rm Sym}^\infty(\Sigma))=\widetilde{{\rm Sym}^k(\Sigma)}$. Give
${\mathcal C}_{*}({\rm Sym}^\infty(\Sigma))$ a multiplicative grading $\mathcal F$ such that
the gradings of $e_{i}$ and $f_{i}$ equal 1, and the grading of $\gamma_{s}$
equals $s$. Then the cellular chain complex of
${\rm Sk}^{k}(\Sigma)$ equals $\oplus_{i\leq k}{\mathcal F}^{i}{\mathcal C}_{*}({\rm Sym}^\infty(\Sigma))$.
\end{theorem}
The more technical algebraic terms used in the statement of this theorem will
be clarified in the subsections below.

\subsection{Recollection of a construction by Kallel and Salvatore.}

This is a general construction
which for any two-dimensional CW-complex $Y$ and any natural number $k$
produces a CW-complex
$ \ReducedSymm k Y$, together with a homotopy equivalence
$r: \SymmPower kY \to \ReducedSymm k Y$.
It uses special properties of cells of dimensions one and two, and the
arguments we shall present do not seem to have an obvious direct generalization 
to higher dimensions.

We will mainly be interested in the special case where
$Y$ is a cell decomposition of a closed oriented  surface $\Sigma$ with one
0-cell, $2g$ 1-cells and one 2-cell. 
It is convenient to 
do the construction in two steps, first treating the 1-skeleton of the surface.
We will denote this 1-skeleton by $X$, so that
\[
X := \mathrm{Sk}^{1}Y \simeq \Sigma \setminus \{  {\rm pt} \}.
\]
Thus $X$ is a one-dimensional CW-complex
with one 0-cell. Chose the basepoint $v_0 \in X$ 
to be the point it represents.  

The 1-cells of $X$ come with characteristic maps
 $\phi_i:D^{1}_i\to X$.
We pick a basepoint in the boundary of each  $D^{1}$. All attaching maps are automatically basepoint-preserving. 
The $k$-fold product $X^k$ has a corresponding cell decomposition;
the cells are indexed by sequences $I=(i_1,\dots,i_k)$, where each
$i_j$ is the index of a 1-cell in $X$. Let $D^I=D^{k}$. 
The characteristic maps of the cells are just products
\[
\phi_I=\prod_{1\leq j\leq k}\phi_{i_j}: D^I\cong \prod_{1\leq j\leq k} D^{1}\to X^k.
\]
These characteristic maps are also basepoint preserving.
The $s$-skeleton of $X^k$ is built inductively by attaching $s$-cells to
the $(s-1)$-skeleton:
\[
{\rm Sk}^s(X^k)={\rm Sk}^{s-1} (X^{k})\cup_{\phi_I} (\cup_I D^I).
\]   
The action of the symmetric group $\SymmGroup k$ on $X^k$ preserves the $s$-skeleton. We obtain
a decomposition
\[
{\rm Sk}^s(X^k)/\SymmGroup k=(X^{k-1}/\SymmGroup k) \cup_\phi (\cup_I D^I/\SymmGroup k).
\]
A free orbit of cells in $X^k$ under $\SymmGroup k$ will contribute a single
cell to this pushout. More generally, if a cell in $X^k$ is built from  $l$ times the 0-cell, and $k-l$ \emph{distinct} 1-cells, its $\SymmGroup k$-orbit will consist of $\binom kl$ cells. Each open cell in this orbit will map 
homeomorphically to 
${\rm Sk}^s(X^k)/\SymmGroup k$.  The union of the corresponding closed cells forms a subspace
$\ReducedSymm k X\subset \SymmPower kX$, which has a CW-structure, with the cells we have just described. The inclusion of this
subspace is a  homotopy equivalence, and there exists a retraction (which 
is also a homotopy equivalence)
\[
r: \SymmPower kX \to \ReducedSymm kX.
\]

This completes the construction for a 1-dimensional CW-complex $X$ with a single 0-cell. The next step is to extend the construction to the case when the $CW$-complex also has 2-cells. Actually, as we have already said, in this paper we are only interested in the cell structure on a closed surface $\Sigma$, formed by gluing a single 2-cell onto a bouquet of circles, and we will specialize to this later.

We have to produce characteristic maps of cells. To generate such maps, we use {\em concatenation}, which is the obvious map
\begin{equation} \label{concat}
\mu_{k,l}:\SymmPower kX \times \SymmPower lX \to \SymmPower {k+l}X,
\end{equation}
and which we also write as $\mu(x,y)=x*y$.
As in the 1-dimensional case, we do not give the symmetric product
itself a cell structure, but we produce a cell complex which
is homotopy equivalent to  the symmetric product. 

Following \cite{KallelSalvatore}, given a 2-complex $Y$ whose 1-skeleton 
$X=\mathrm{Sk}^1Y$ is a CW complex with a single 0-cell, we define   
$\ReducedSymm kY$ as the quotient of $\SymmPower kY$ which for
$x, x^\prime\in \SymmPower lX$ and $y \in \SymmPower {k-l}Y$ identifies
$x * y$ with $x^\prime * y$, if $x$ and $x^\prime$ have the same image in
$\ReducedSymm lX$. The quotient map 
\begin{equation} \label{hmtp}
r:\SymmPower kY \to \ReducedSymm kY
\end{equation}
is a homotopy equivalence. In this case, there is no
corresponding inclusion of $\ReducedSymm kY$ into $\SymmPower kY$, so $r$ 
is not a deformation retraction to a subspace.

The  retraction  (\ref{hmtp}) is compatible with concatenation (\ref{concat}), 
so there are commutative diagrams
\[
\xymatrix{
\SymmPower kY \times \SymmPower lY\ar[d]^{(r,r)}\ar[r]^-{\mu_{k,l}}
& \SymmPower {k+l}Y\ar[d]^r\\
{\ReducedSymm kY\times \ReducedSymm lY} \ar[r]^-{\bar \mu_{k,l}}&\ReducedSymm {k,l}Y.\\
}
\]

The attaching maps of cells are given in the following fashion. The cells in 
$\ReducedSymm kX= \ReducedSymm k{{\mathrm{Sk}^1Y}}$ are products of distinct 1-cells 
\[
D^k \cong (D^1)^k \to X^k = (\SymmPower 1 X)^k \xrightarrow{\mu} \SymmPower k X
\xrightarrow{r} \ReducedSymm kX.
\]
There is also the 2-cell in $Y=X\cup e^2$ with characteristic map
$\gamma_2 : D^2 \to Y$. We can assume that $\gamma_2$ is basepoint-preserving. 
Using concatenation, it produces cells of all even dimensions 
\[
\gamma_{2j} : D^{2j} \cong (D^2)^j \xrightarrow{(\gamma_2)^j}
Y^k \xrightarrow{\mu} Y. 
\]
Finally, we can concatenate the cells in $\ReducedSymm{{k-j}}X$ with $\gamma_j$ for $1 \le j \le g$. This produces
a complete list of cells in a cell decomposition of $\ReducedSymm k Y$. Note that every
characteristic map is basepoint-preserving.
The homology of $\ReducedSymm k Y$ can be computed as the homology of the corresponding
cellular chain complex $C_*(\ReducedSymm k Y)$, whose chain group $C_n(\ReducedSymm k Y)$
has a basis indexed by the $n$-cells (see for instance Chapter 2.2 of \cite{Hatcher} for the
definition of chain complexes and a discussion of cellular homology).

We note that concatenation induces a commutative product  
\[
\overline{\mu_{k,l}}:
\ReducedSymm kY \times \ReducedSymm lY\to
\ReducedSymm {k+l}Y.
\]
In particular, multiplication with the basepoint
$v_0 \in \ReducedSymm kY$ 
induces stabilization maps 
$
i:\ReducedSymm kY\to
\ReducedSymm{k+l}Y.
$
 By the description of the 
cells of $\ReducedSymm k Y$, this map identifies
$\ReducedSymm kY$ with a subcomplex of $\ReducedSymm {k+l}Y$.
It follows that the induced map of chain groups
$i_*: C_*(\ReducedSymm k Y) \to  C_*(\ReducedSymm {k+l} Y)$
is a split injective map.

The multiplication passes to a graded commutative product of chain groups
\[
\overline{\mu_{k,l*}}:
C_m(\ReducedSymm k\Sigma)\otimes 
C_n(\ReducedSymm l\Sigma)\to
C_{m+n}(\ReducedSymm{k+l}\Sigma).
\]
Here, graded commutative means that 
$\mu_{k,l*}(a\otimes b) =(-1)^{mn}\mu_{l,k*}(b\otimes a)$.
Moreover, the differential in the cellular chain complex is a 
derivation with respect to this product. That is,  
if $a\in C_m(\ReducedSymm k Y)$ and
$b\in C_n (\ReducedSymm l Y)$, then
$\partial \mu_{k,l*}(a\otimes b)=
\mu_{k,l*}(\partial a\otimes b)+
(-1)^k \mu_{k,l*}(a\otimes \partial b)$.

This product exhibits $\colim{k} C_*(\ReducedSymm k Y)$ as a graded
commutative ring, containing $C_*(\ReducedSymm k Y)$ as a
direct summand. Note that this summand is not closed under ring multiplication.
Since the boundary map is a derivation,  
$\colim k C_*(\ReducedSymm k Y)$ is a differential graded algebra (DGA), i.e.\ a graded algebra
together with a chain complex structure whose differential preserves the algebraic structure (see
\cite{GelMan}, Chapter V.3).

\begin{example}
If $S^1$ is the CW-complex with one 0-cell and one 1-cell $e_1$,
then for any $k\geq 1$ the stabilization map 
\[
S^1=\ReducedSymm 1{S^1} \to \ReducedSymm k{S^1} 
\]  
is a homeomorphism. The ring $\lim_k C_*(\ReducedSymm k {S^1})$
equals the exterior algebra $\Lambda[e_1]=\ZZ[e_1]/e_1^2=0$. The differential is trivial,
and for $k\geq 1$ the inclusion map $C_*(\ReducedSymm k {S^1})\subset\colim k C_*(\ReducedSymm k {S^1})$ is an isomorphism.  
\end{example}
\begin{example}
\label{ex:1-skeleton}
Let $X$ be a 1-dimensional CW-complex with one 0-cell and 1-cells $e_1,\dots, e_n$. The inclusion of
the cells collectively define an isomorphism  
$\Lambda[e_1,\dots, e_n]\to \colim k C_*(\ReducedSymm kX)$. 
The differential is trivial, and under this isomorphism the
image of the inclusion $C_*(\ReducedSymm kX) \hookrightarrow \colim k C_*(\ReducedSymm kX)$ consists
of monomials without repetitions in the generators with $k$ or less factors. In particular, if
$k\geq n$, this inclusion is an isomorphism. 
\end{example}

We now specialize to the case we really are interested in. So let $Y$ be a cell decomposition 
of a surface $\Sigma_g$ of genus $g$, with one 0-cell $v_0$, $2g$ 1-cells and one 2-cell. 
We choose the 1-cells in sets of pairs $e_i$,$f_i$, $1\leq i\leq g$,
and let the 2-cell $D^2$  be attached by the word $\prod_{1\leq i\leq g}(e_if_ie_i^{-1}f_i^{-1})$.
We can assume that the loop $e_i$ represents $x_i$ (considered as class
in $H_1(\Sigma; \mathbb{Z})$) and that the loop $f_i$ represents $y_i$ (see Section~\ref{topSym}). 
We also chose a basepoint in the boundary of each standard
disk $D^n$ ($n=1,2$). 

Let $\gamma_2:D^2 \to Y \to \colim k C_*(\ReducedSymm kX)$ be the
characteristic map of the 2-cell. Taking powers of this map, and
dividing by the action of the symmetric group, we obtain 
maps
\[
\gamma_{2n}:D^{2n} \cong (D^2)^k \xrightarrow{\gamma_2^n} \ReducedSymm nX.
\] 
We abusively denote the corresponding cycles by $\gamma_{2n}$.

According to Example \ref{ex:1-skeleton}, the inclusion 
$X=\mathrm{Sk}^1(Y)\subset Y$ induces an inclusion
\[
\Lambda[e_1,\dots,e_g,f_1,\dots,f_g]\hookrightarrow
\colim k C_*(\ReducedSymm kX).
\]
This makes  $\colim k C_*(\ReducedSymm kY)$ into a free module over
$\Lambda[e_1,\dots,e_g,f_1,\dots,f_g]$ with generators
$\gamma_{2n}$, $n\geq 0$. The symmetrising map
\[
(D^{2})^k \to (D^{2})^k/\SymmGroup k \cong D^{2k}
\]
has degree $k!$. It follows that 
$(\gamma_2)^k=k!\gamma_{2k}$, and consequently that
\begin{equation}
\label{eq:gammaAlgebra}
\gamma_{2k}\gamma_{2l}= {{k+l}\choose k} \gamma_{2(k+l)}.
\end{equation}

The relation (\ref{eq:gammaAlgebra}) means that the characteristic map $\gamma_2$ of the 2-cell
generates a subalgebra in the cellular chain complex that is isomorphic to what is called a
{\em gamma-algebra}, and which we will denote by $\Gamma[\gamma_2]$ following general practice.
More concretely, 
as a group $\Gamma[\gamma_2]$ is the free Abelian group generated 
by the classes $\gamma_{2n}$ for $n\in \mathbb{N}_0$. This comes with a grading, generated 
additively from 
the prescription that  $\gamma_2$ has degree two; so $\gamma_{2n}$ has degree $2n$.  
The multiplication is
defined on generators exactly by the relation
(\ref{eq:gammaAlgebra}).

It follows from the definitions that
\[
\colim k C_*(\ReducedSymm kY)\cong \Lambda[e_1,\dots e_g,f_1,\dots,f_g]
\otimes \Gamma[\gamma_2].
\]

The differential on $e_1$ is trivial by naturality. Since
$\gamma_1$ is the top cell in $\Sigma_g$, the differential on $\gamma_1$ is zero.
It follows from the Leibniz rule that that 
$\partial (k!\gamma_k)=\partial (\gamma_1^k)
=k!(\gamma_1)^{k-1}\partial \gamma_1=0$. Since the chain groups are free 
groups, it follows that $\partial \gamma_k=0$, and thus that
the differential is trivial. If we introduce the grading
$\deg(e_i)=\deg(f_j)=1,\deg(\gamma_i)=i$,
the image
$C_*(\ReducedSymm kY)\subset\lim_k C_*(\ReducedSymm k Y)$
exactly consists of the classes in grading $\leq k$. 
 
Because the differential is trivial, the chain groups agree with the homology groups.
This completes the computation in \cite{KallelSalvatore} of the homology of
$\SymmPower k\Sigma$.

\subsection{The chain complex of the universal covering space.}

We now consider the universal cover of $\SymmPower kY$ for $k\geq 2$.
As before, fix  a standard basis  $\{x_1,\dots ,x_g,y_1,\dots y_g\} $ as in (\ref{symplbasis})
for
\[ 
\pi := \pi_1({\rm Sym}^k(\Sigma))\cong H_{1}(\Sigma;\mathbb{Z}).
\]
The cell decomposition of $\SymmPower kY$ studied in the previous subsection
lifts to a cell decomposition of 
the universal cover
\[
U_k(Y)=\widetilde{\ReducedSymm kY}
\]
 for $k\geq 2$.
Each cell of $\ReducedSymm kY$ lifts to a free $\FundGroup$-orbit of
cells in the universal cover. 
In particular, the basepoint $v_0\in \ReducedSymm kY$ is covered by 
a free $\pi$-orbit of  points.
Pick one of the points in this orbit to be the basepoint of $U_k(Y)$.   

For $k\in \{ 0,1 \}$ we similarly consider the covering
$U_k(Y)$ of
$\ReducedSymm kY$ corresponding to $\FundGroup$.
So $U_0(Y)$ equals $\FundGroup$ as a discrete set, and 
$U_1(Y)$ is the simply connected covering of
$Y$ determined by the kernel of the Hurewicz map. 
Note that there is a short exact
sequence of  monoids:
\[
\FundGroup \to \coprod_{k\geq 0}  U_k(Y) \to
\coprod_{k\geq 0} {\ReducedSymm kY}.
\]
We lift the basepoint to $0\in \FundGroup$. The cells of 
${\ReducedSymm kY}$ lift to cells of  
$U_n(Y)$. More precisely, the characteristic maps of the cells 
in ${\ReducedSymm kY}$ lifts to character maps of cells in $U_n(Y)$.

Recall that each cell of positive dimension contains the basepoint 
in its boundary, and there is a unique basepoint preserving lifting of
the characteristic map.  By abuse of notation, we use the
same notation for a based cell and its based lifting.
We obtain all cells in $U_n(Y)$ by 
translating these cells by the elements of $\FundGroup$. It follows that
we have a CW-structure on $U_n(\Sigma)$
with cells 
\[
\alpha *v_0^k*e_{i_1}*\dots * e_{i_l}*f_{j_1}*\dots f_{j_m}*\gamma_s,
\] 
where $\alpha\in \FundGroup$ and $k+j+m+s=n$. 
The cellular chain groups are
\[
\ZZ[x_1^{\pm 1},\dots x_g^{\pm 1}, y_1^{\pm 1},\dots y_g^{\pm 1}]\otimes
\Lambda[e_1,\dots,e_g,f_1,\dots,f_g]\otimes \Gamma[\gamma].
\]
The differentials are derivations. Since $e_i$ is the path
representing $x_i$, we have that $\partial e_i=x_i-1$, and similarly
$\partial f_i=y_i-1$. The differential of $\gamma_1$ is given by tracing the boundary
of the lifted cell. For concreteness, we draw the lifted cell 
for the particular case of a genus two surface including the names of all cells in its boundary.
The zero cell $gv_0$ is denoted by the corresponding group element $g$: 
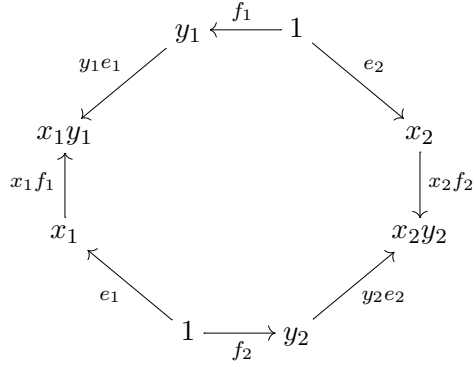
\[
\xymatrix{
&y_1\ar[ld]_{y_1e_1}&1\ar[l]_{f_1}\ar[rd]^{e_2}&\\
x_1y_1&&&x_2\ar[d]^{x_2f_2}\\
x_1\ar[u]^{x_1f_1}&&&x_2y_2\\
&1\ar[lu]^{e_1}\ar[r]_{f_2}&y_2\ar[ru]_{y_2e_2}&\\
}
\] 

We obtain that 
\[
\partial \gamma_1=\sum_i(e_i + x_if_i-y_ie_i-f_i)=
\sum_i((1-y_i)e_i-(1-x_i)f_i),
\]
and since $\partial$ is a derivation, that 
\[
\partial\gamma_s=\sum_i((1-y_i)e_i-(1-x_i)f_i)\gamma_{s-1}.
\]
Let us write $\lambda:=\sum_i((1-y_i)e_i-(1-x_i)f_i)$, so that this equation simplifies to
\begin{equation}
\label{eq:boundary}
\partial\gamma_s=\lambda\gamma_{s-1}.
\end{equation}

\section{Homology computations}
\label{sec:Homology}

In this section and the next, $X_n$ will denote a one-dimensional CW-complex with $n$ one-cells;
$Y_g$ will denote a two-dimensional complex $Y_g:=X_{2g}\cup e_{2}$ such that $Y_g$  
is homotopy equivalent to the closed Riemann surface of genus $g$ that we call $\Sigma_g$. That is, $Y_{g}$ is a choice of a CW-structure on
$\Sigma_g$.

The universal cover of the symmetric product
$\SymmPower k\Sigma$ which we denote $\ReducedSymm k{\Sigma_g}$, 
is homotopy equivalent to  the universal cover $\ReducedSymm k{Y_g}$.
Restricting this universal cover to the  symmetric product of the
one skeleton $X_{2g}$ of $Y_g$, we obtain a (not universal) covering space $U_k(X_{2g})$ of
$\ReducedSymm k{X_{2g}}$. The total space of this covering has a CW structure as
a subcomplex of $U_k(Y_{g})$. 
In Theorem \ref{th:chain_complex} we used this CW structure 
to determine a chain complex whose
homology equals the homology of $U_k(Y_{g})$. In this section, we shall
compute the homology of that chain complex.

The strategy of the computation
will be as follows. We first compute the homology of the $U_{k}(X_{2g})$;
this is done in Lemma \ref{le:torus}. Then we relate
$H_{*}(U_{k}(Y_{g}))$ to the homology of $U_{k}(X_{2g})$ in Lemma
\ref{le:reduction_to_torus}. In order to complete the calculation,
we have to do some homological algebra. The main step of this
homological calculation is collected in Lemma \ref{le:Cohomology}, before
stating and proving our main Theorem \ref{th:main}.
(This calculation can also be formulated using spectral sequences, but
we prefer to keep the formal machinery as elementary as we can.)

Let $X=X_{n}$ be a wedge of $n$ copies of the circle.
The cellular homology complex corresponding to the cell decomposition
of the symmetric product that  
we discussed in Section \ref{sec:CW}
can, as explained there, be identified with the DGA
\[
\Chains_{*}(X)= \ZZ[z_{1}^{\pm 1},\dots, z_{n}^{\pm 1}]\otimes \Lambda[\psi_1,\dots,\psi_{n}],\qquad 
\partial \psi_{i}=z_{i}-1.
\]  
Here $\ZZ[z_{1}^\pm,\dots z_{n}^\pm]=\ZZ[\FundGroup]$, and $\Lambda[\psi_1,\dots,\psi_{n}]$
is the exterior algebra on the generators $\psi_{i}$, each of degree 1. 

This chain complex computes the homology of 
$\colim k C_*(U_k(X_{n}))$. To obtain the homology of 
$C_*(U_{k}(X_{n}))$, we have to take the subcomplex corresponding to
classes of grading $\leq k$. Fortunately, the grading agrees with the
internal grading, so the subcomplex  $\Chains_{*}(U_k(X_{n}))\subset \Chains_*(X_n)$ is the truncated complex
\[
\Chains_{k}(X_{n})\xrightarrow{\partial_{k}} \dots \Chains_{1}(X_{n})
\xrightarrow{\partial_{1}} \Chains_{0}(X_{n}).
\] 
We will use this complex to compute the homology of $U_k(X_{n})$.

\begin{definition}
 $\Cycles_{k}:= \ker\{\partial_{k}: \Chains_{k}(X_n)\to \Chains_{k-1}(X_n)\}$.
\end{definition}

\begin{remark}
\label{rem:chains}
If $0 \leq k < n$, $\Cycles_{k}$ is not finitely generated as an Abelian
group. 
In fact, $\Chains_k(X_{n})$ contains no $\ZZ[\FundGroup]$-torsion, so that 
any nontrivial element of $\Cycles_k\subset \Chains_k(X_{n})$ 
generates a free  $\ZZ[\FundGroup]$-submodule of rank one.
As an Abelian group, this submodule is free on infinitely many generators.
An example of a nontrivial element in $\Cycles_k$ for $0\leq k<n$ 
is $\partial(\psi_1\psi_2\dots \psi_{k+1})$.
\end{remark}

\begin{lemma}
\label{le:torus}
Assume that $k\geq 2$.
\[
H_i(U_k(X_{n}))=
\begin{cases}
\ZZ &\text{if $i=0$,}\\
\Cycles_k&\text{if $0<i=k\leq n$,}\\
0&\text{else.} 
\end{cases}
\]
\begin{proof}
The untruncated chain complex
\[
0\to \Chains_{n}(X_{n})\xrightarrow{\partial_{k}} \dots 
\Chains_{1}(X_{n})\xrightarrow{\partial_{1}} \Chains_{0}(X_{n})\to 0
\] 
is the tensor product of $n$ copies of the chain complex
\[
\ZZ[z^{\pm 1}]\otimes \Lambda[\psi], \qquad \partial \psi=z-1
\]

The homology of this chain complex is $\ZZ$, generated by the class 
1 in dimension 0, so by K\"u{}nneth's formula (see \cite{MacLane} Chapter V, Theorem 10.1), the homology of $\Chains_{*}(X_{n})$ also consists of
$\ZZ$ in dimension 0. The lemma follows on comparing this with the
homology of the truncated chain complex.  
\end{proof}
\end{lemma}

The cellular chain complex $\Chains^k_*(Y_{g})$ of 
$U_{k}(Y_{g})$ is the subcomplex of

\[
 \Gamma[\gamma_2] \otimes 
\ZZ[x_1^{\pm 1},\dots x_g^{\pm 1}, y_1^{\pm 1},\dots y_g^{\pm 1}]\otimes
\Lambda[e_1,\dots,e_g,f_1,\dots,f_g]=
 \Gamma[\gamma_2] \otimes \Chains_*(X_{2g})
\]
generated by cells of grading $\leq k$, where the grading is
multiplicative on the cells, the gradings of $x_{i}$ and $y_{i}$ are zero,
the gradings of $e_{i}$ and $f_{i}$ are 1, and the grading of each $\gamma_{2i}$
is $i$. This time, the grading does not agree with the internal grading
because the internal grading of $\gamma_{2i}$ is $2i$, not $i$.

There is a natural chain of CW-complexes
\[
\{ {\rm pt} \} =U_0(\Sigma_g)\subset U_{1}(Y_{g})\subset 
U_{2}(Y_{g})\subset
\dots \subset 
U_{k}(Y_{g})\subset \dots 
\]
Each inclusion is the inclusion of a subcomplex. 
There is a corresponding sequence of cellular chain complexes.
Let $Q_{k}=\Chains_{*}U_{k}(Y_{g})/\Chains_{*}U_{k-1}(Y_{g})$
be the quotient chain complex. We consider $\Chains_*(X_{2g})$ as a subcomplex of $\Chains_*(Y_g)$, using the identifications $x_i=\psi_i$, $y_j=\psi_{g+j}$ for $1\leq i,j \leq g$.

\begin{lemma}
\label{le:Q}
The homology of the chain complex $Q_{k}$ is concentrated in
one dimension.
More precisely
\[
H_{i}(Q_{k})\cong
\begin{cases}
\Chains_{k}(X_{2g})/\lambda\Chains_{k}(X_{2g}) &\text{if $k\leq 2g$ and $i=k$,} \\
\ZZ \text{ generated by $\gamma_{k-2g}$}&\text{if $k\geq 2g$ and $i=2k-2g$},\\
0 &\text{else.}
\end{cases}
\]
\begin{proof}
The chain complex $Q_{k}$ is given by
\[
0\xrightarrow{}
\gamma_{k}\Chains_{0}(X_{2g})\xrightarrow{\partial}
\gamma_{k-1}\Chains_{1}(X_{2g})\xrightarrow{\partial}\dots
\xrightarrow{\partial} 
\gamma_{1}\Chains_{k-1}(X_{2g})\xrightarrow{\partial}
\Chains_{k}(X_{2g})\xrightarrow {}0.
\]
Because $\partial \gamma_{i}=\lambda \gamma_{i-1}$, this chain complex
is isomorphic to the chain complex
\begin{equation}
\label{eq:truncated}
0 \xrightarrow{}
\Chains_{0}(X_{2g})\xrightarrow{\lambda}
\Chains_{1}(X_{2g})\xrightarrow{\lambda}\dots
\xrightarrow{\lambda} 
\Chains_{k-1}(X_{2g})\xrightarrow{\lambda}
\Chains_{k}(X_{2g})\xrightarrow {}0.
\end{equation}
We first consider the corresponding complex
\[
0\xrightarrow{}
{\Chains}_{0}(X_{2g}) \xrightarrow{\lambda}
{\Chains}_{1}(X_{2g}) \xrightarrow{\lambda}
\dots
{\Chains}_{2g-1}(X_{2g}) \xrightarrow{\lambda}
{\Chains}_{2g}(X_{2g})\xrightarrow{}
0,
\]
and claim that its cohomology
is 0 except in degree $2g$, where the cohomology is $\ZZ$.

We can write the complex ${\Chains}_{*}(X_{2g})$ as a tensor product
\[
\otimes_{1\leq i\leq g}(\ZZ[x_i^{\pm 1}]\otimes \Lambda[f_i])
\otimes_{1\leq j\leq g}(\ZZ[y_j]^{\pm 1}\otimes \Lambda[e_j])
\]
where the boundary maps are multiplication by $-(1-x_i)f_{i}$ in
$\ZZ[x_i^{\pm 1}]\otimes \Lambda[f_i]$, respectively multiplication by
$(1-y_{j})e_{j}$ in $\ZZ[y_j^{\pm 1}]\otimes \Lambda[e_j]$. The cohomology of
$\ZZ[x_i]^{\pm 1}\otimes \Lambda[f_i]$ is a copy of $\ZZ$ generated by
$f_i$, and the cohomology of 
$\ZZ[y_i]^{\pm 1}\otimes \Lambda[e_i]$ is a copy of $\ZZ$ generated by
$e_i$; actually, there is a short exact sequence
\begin{align*}
0 \to \ZZ[x_i^{\pm 1}]\xrightarrow{(x_i-1)f_i}
f_i\ZZ[x_i^{\pm 1}]&\to\ZZ\to 0.\\
f_i\sum_k{n_k x_i^k}&\mapsto \sum_k{n_k}\\
\end{align*}
It follows from the K\"u{}nneth formula that the cohomology of
${\Chains}_{*}$ is a single copy of $\ZZ$ generated by 
$(\prod_{1\leq j\leq  g}{e_j})(\prod_{1\leq i\leq g}f_i)$.

We now return to the truncated chain complex (\ref{eq:truncated}).
This chain complex is exact, except at $\Chains_{k}(X_{2g})$ if
$k<2g$ and at $\Chains_{2g}(X_{2g})$ if $k \geq 2g$.
The lemma follows.
\end{proof}
\end{lemma}

\begin{lemma}
\label{le:reduction_to_torus}
Let 
\[
\partial_{k}:\Chains_{k}(X_{2g})/\lambda\Chains_{k-1}(X_{2g})
\to
\Chains_{k-1}(X_{2g})/\lambda\Chains_{k-2}(X_{2g})
\]
be the boundary map.  
\[
H_{i}(U_{k}(Y_{g}))\cong
\begin{cases}
\ZZ & \text{if $i$ is even, $k > 2g$  and $i\leq 2k-2g$.}\\
\ker \partial_{k}& \text{ if $i=k$ and $k\leq 2g$,}\\
\ker \partial_{i}/ \im \,\partial_{i-1}& \text{if $i<k\leq 2g$}\\
0&\text{ else.} 
\end{cases}
\] 
\begin{proof}
The long exact sequence
\[
\dots\to
H_i(\Chains_{*}U_{k-1}(Y_{g})) \to
H_i(\Chains_{*}U_{k}(Y_{g}))\to
H_{i}(Q_{k}) \to
\dots
\]
breaks up into shorter exact sequences, since $H_{i}(Q_{k})$ is only nontrivial in a single
dimension. We need to consider several cases separately.

 When $k\leq 2g$, we have isomorphisms
\begin{equation}
\label{eq:iso}
H_i(\Chains_{*}U_{k-1}(Y_{g})) \to
H_i(\Chains_{*}U_{k}(Y_{g})), \quad i\not=k-1,k,
\end{equation}
together with the exact sequence
\begin{equation}
\label{eq:exact}
\xymatrix{
0 \ar[r] & 
H_k(\Chains_{*}U_{k-1}(Y_{g}))\ar[r] &
H_k(\Chains_{*}U_{k}(Y_{g}))\ar[r]&
H_{k}(Q_{k})\ar  `r[d] `[ll] `^d[dlll]`^r[d] [dll] \\
&
H_{k-1}(\Chains_{*}U_{k-1}(Y_{g}))\ar[r] &
H_{k-1}(\Chains_{*}U_{k}(Y_{g}))\ar[r]&  0.
}
\end{equation}

\noindent
(I.1)
\emph{The case $k\leq 2g$ and $i >k$:}
If $i>2k$, then $\Chains_i(U_k(Y_g))=0$, so certainly 
$H_i(\Chains_*U_k(Y_g))=0$.
It follows by downwards induction from this and from (\ref{eq:iso}) 
that if $k\leq 2g$ and $i>k$, then 
$H_i(\Chains_{*}U_{k}(Y_{g}))=0$.

\noindent
(I.2)
\emph{The case $k\leq 2g$ and $i =k$:}
It follows from case (I.1) that $H_k(\Chains_{*}U_{k-1}(Y_{g}))\cong 0$, so that 
 by~(\ref{eq:exact}) the group
$H_k(\Chains_*U_k(Y_g))$ is isomorphic to the kernel of the map
$H_k(Q_k)\to H_{k-1}(\Chains_*U_{k-1}(Y_g))$.

Moreover, it also follows from~(\ref{eq:exact}) that
$H_k(\Chains_{*}U_{k}(Y_{g}))\to H_k(Q_k)$ and similarly 
 $H_{k-1}(\Chains_{*}U_{k-1}(Y_{g}))\to H_{k-1}(Q_{k-1})$
are injections. Putting these fact together, we obtain that 
if $k\leq 2g$, then 
$H_k(\Chains_{*}U_{k}(Y_{g}))$ is isomorphic
to the kernel of the composite map
\[
\partial_{k}:H_{k}(Q_{k}) \to H_{k-1}(\Chains_{*}U_{k-1}(Y_{g})) 
\hookrightarrow H_{k-1}(Q_{k-1}).
\]
Using lemma~\ref{le:Q}  completes the proof of the lemma for these values of $i,g$ and $k$.

\noindent
(I.3)
\emph{The case $k \leq 2g$ and $i<k$:}
Because of (\ref{eq:iso}) it is sufficient to
compute $H_{k-1}(\Chains_{*}U_{k}(Y_{g}))$. This 
can be read off from the following diagram with exact rows:
\[
\xymatrix{
H_{k}(Q_{k})\ar[r]\ar@{=}[d]&H_{k-1}(\Chains_{*}U_{k-1}(Y_{g}))\ar[r]\ar[d]^{\cong}
&H_{k-1}(\Chains_{*}U_{k}(Y_g))\ar[r]\ar[d]^{\cong}&0\\
H_{k}(Q_{k})\ar[r]\ar[rd]^{\partial_{k}}&\ker \partial_{k-1}\ar[r]\ar@{^{(}->}[d]&\ker \partial_{k-1}/\im
\, \partial_{k}\ar[r]& 0\\
&H_{k-1}(Q_{k-1}).&&
}
\]

Now we turn to the cases when $k>2g$. The lemma claims that under this assumption
\[
H_i(U_k(Y_g))=
\begin{cases}
\ZZ &\text{ if $i$ is even and $i\leq 2k-2g$}\\
0 &\text{ else.}
\end{cases}
\]
Using lemma~\ref{le:Q} and the short exact sequence of chain complexes 
$0\to \Chains_*U_{k-1}(Y_g)\to \Chains_*U_k(Y_g)\to Q_k\to 0$ we see that
for $k > 2g$, we have isomorphisms
\begin{equation}
\label{eq:highiso}
H_{i}(\Chains_{*}U_{k-1}(Y_{g}))\cong H_{i}(\Chains_{*}U_{k}(Y_{g})), \quad i\not=2k-2g-1,2k-2g.
\end{equation}

\noindent
(II.1)
\emph{The case $k>2g$ and $i \geq 2k-2g+1$:} 
In this case, $i>2g$. It follows from case (I.1) above that 
$H_{i}(\Chains_{*}U_{2g}(Y_{g}))\cong 0$, so we obtain from
(\ref{eq:highiso}) that
\[
0 \cong H_{i}(\Chains_{*}U_{2g}(Y_{g})) \cong
H_{i}(\Chains_{*}U_{2g+1}(Y_{g}))\cong \dots 
\cong H_{i}(\Chains_{*}U_{k}(Y_{g})),
\]
and consequently $H_{i}(\Chains_{*}U_{k}(Y_{g}))\cong 0$. 

\noindent
(II.2)
\emph{ The case $k>2g$ and $2k-2g-1\leq i\leq 2k-2g$:} 
We know from case (II.1) that
\begin{align*}
H_{2k-2g}((\Chains_{*}U_{k-1}(Y_{g})))&\cong 0,\\
H_{2k-2g-1}((\Chains_{*}U_{k-1}(Y_{g})))&\cong 0.\\
\end{align*}
From the short exact sequence 
$0\to \Chains_*U_{k-1}(Y_g)\to \Chains_*U_k(Y_g)\to Q_k\to 0$
together with lemma~(\ref{le:Q}), we conclude
\begin{align*}
H_{2k-2g}((\Chains_{*}U_{k}(Y_{g})))&\cong H_{2k-2g}(Q_{k})\cong \ZZ,\\
H_{2k-2g-1}((\Chains_{*}U_{k}(Y_{g})))&\cong 0.\\
\end{align*}

\noindent
(II.3)
\emph{ The case $k>2g$, $i\leq 2k-2g-2$:} This follows from 
case II.2 together with the isomorphisms~(\ref{eq:highiso}).
\end{proof}
\end{lemma}

Before we prove our main theorem we need one more technical lemma.
For $1\le m \le g$, let 
\[
\sigma_{m}:=\sum_{1\leq i_{1}<i_{2}<\dots <i_{m}\leq g} 
e_{i_{1}}f_{i_{1}}\dots e_{i_{m}}f_{i_{m}}
\]
be the
elementary symmetric polynomials in the $g$ products $e_{i}f_{i}$.
Note that $\lambda=-\partial\sigma_{1}$ where $\lambda$ was defined before equation  (\ref{eq:boundary}).
\begin{lemma}
\label{le:Cohomology}
The cohomology of the complex 
\[
\Cycles_{0} \xrightarrow{\lambda}
\Cycles_{1} \xrightarrow{\lambda}
\dots
\Cycles_{2g-1} \xrightarrow{\lambda}
\Cycles_{2g} 
\]
is $\ZZ$ in odd degrees degrees $2m+1$, where $2g-1\geq 2m+1\geq 3$, and 0 elsewhere.
$H^{2m+1}(\Cycles_*)$ is generated by $\lambda\sigma_{m}$.
\begin{proof}
We consider the complex
\[
{\Chains}_{0}(X_{2g}) \xrightarrow{\lambda}
{\Chains}_{1}(X_{2g}) \xrightarrow{\lambda}
\dots
{\Chains}_{2g-1}(X_{2g}) \xrightarrow{\lambda}
{\Chains}_{2g}(X_{2g}) .
\]
Let ${\mathbf B}_0$ be the image of 
$\partial:{\Chains}_1(X_{2g})  \to {\Chains}_0(X_{2g})$. There is a short exact sequence of
cochain complexes:
\[
\xymatrix{
0\ar[r]&{\mathbf B}_0 \ar[r]^{i_0}\ar[d]&{\Chains}_{0}(X_{2g})\ar[r]^{q_0}\ar[d]^{\lambda}&\ZZ\ar[r]\ar[d]^0&0\\
0\ar[r]&{\mathbf K}_1 \ar[r]^{i_1}\ar[d]^{\lambda}&{\Chains}_{1}(X_{2g})\ar[r]^{q_1}\ar[d]^{\lambda}&{\mathbf B}_0\ar[r]\ar[d]&0\\ 
0\ar[r]&\Cycles_2 \ar[r]^{i_2}\ar[d]&{\Chains}_{2}(X_{2g})\ar[r]^{q_2}\ar[d]^{\lambda}&\Cycles_1\ar[r]\ar[d]&0\\
\vdots&\vdots&\vdots&\vdots&\vdots\\
0\ar[r]&\Cycles_{2g} \ar[r]^{i_*}&{\Chains}_{2g}(X_{2g})\ar[r]^{q_*}&\Cycles_{2g-1}\ar[r]&0\\
}
\]
Let ${\mathbf H}^i$ be the cohomology of the cochain complex
\begin{equation}
\label{eq:Bcomplex}
0\to {\mathbf B}_0 \to \Cycles_1 \xrightarrow{\lambda} \Cycles_2 \to \dots
\to{\Cycles}_{2g} \to 0 
\end{equation}
The short exact sequence of  complexes above induces a long exact
sequence of cohomology groups,

\begin{equation}
\label{LES}
\xymatrix{
&{\mathbf H}^0 \ar[r]&0\ar[r]&\ZZ
\ar `r[d] `[ll] `^d[dlll]`^r[d] [dll]
\\
&{\mathbf H}^1 \ar[r]&0\ar[r]&{\mathbf H}^0
\ar `r[d] `[ll] `^d[dlll]`^r[d] [dll]
\\
&{\mathbf H}^2 \ar[r]&0\ar[r]&{\mathbf H}^1\\
&\vdots&\vdots&\vdots\\
&{\mathbf H}^{2g} \ar[r]&\ZZ\ar[r]&{\mathbf H}^{2g-1}\\
}
\end{equation}
It follows inductively that
\[
{\mathbf H}^i=
\begin{cases}
\ZZ& \text{if $1\leq i\leq 2g-1$, and $i$ even,}\\
0&\text{else.} 
\end{cases}
\]
 Actually, we can find a generator for each of the homology groups
inductively. 
The first claim is that ${\mathbf H}^{1}$ is
generated by $\lambda$. This follows from 
the long exact sequence (\ref{LES}), by the definition of the boundary map.
(see e.g.\ \cite{Hatcher} Chapter 2).
Assume inductively that 
${\mathbf H}^{2m-1}$ is generated by 
$\lambda\sigma_{m-1}$. The group ${\mathbf H}^{2m+1}$ is generated by the image of
$\lambda\sigma_{m-1}$ under the boundary map of the long exact sequence of
(\ref{LES}). To compute this boundary, we first find a chain in 
${\Chains}_{2g-1}(X_{2g})$ that maps to the cycle $\lambda \sigma_{m-1}$. 
$\sigma_m$ is up to sign such a chain:
\begin{align*}
\partial \sigma_{m}&=
\partial \sum_{1\leq i_{1}<i_{2}<\dots <i_{m}\leq g} 
e_{i_{1}}f_{i_{1}}\dots e_{i_{m}}f_{i_{m}}\\
&=-\lambda
\sum_{1\leq i_{1}<i_{2}<\dots <i_{m-1}\leq g} 
e_{i_{1}}f_{i_{1}}\dots e_{i_{m-1}}f_{i_{m-1}}\\
&=-\lambda \sigma_{m-1}.
\end{align*}
By the definition of the boundary map,  
${\mathbf H}^{2m+1}$ is generated by the product of $\lambda$
with $\sigma_m$, which completes the induction.

This finishes the computation of the cohomology of the cochain complex
(\ref{eq:Bcomplex}). This cohomology agrees with the cohomology of 
$\Cycles_*$ except in degrees 0 and 1.
To complete the proof of the lemma, we now only have to check 
that it is true for those two dimensions as well. 
But this also follows from our computation, and the fact that
since $\Cycles_0=\Chains_0(X_{2g})$, there is a short exact sequence
\[
0\to {\mathbf B}_0 \to \Cycles_0 \to \ZZ \to 0.
\]
\end{proof}
\end{lemma}

At this point we are able to prove the main theorem of this section.

\begin{theorem}
\label{th:main}
Let $k\geq 2$. Then 
\[
H_{i}(\widetilde{\SymmPower {k}\Sigma})\cong
H_i(U_k(Y))\cong 
\begin{cases}
\ZZ &\text{$i$ even, $i\not=k$ and $i\leq \max(2k-2g,k)$,}\\
\Cycles_{k}/\lambda \Cycles_{k-1} \oplus \ZZ& \text{ $i=k$ even, $k\leq 2g$,}\\
\lambda \Cycles_{k}& \text{ $i=k$ odd, $k\leq 2g$}\\
0 & \text{else}. 
\end{cases}
\]
\begin{proof}
In the notations of Lemma~\ref{le:reduction_to_torus} 
we have to compute the groups $\ker \partial_k$ and
$\ker \partial_k/\im \partial_{k-1}$. 

Consider the following commutative diagram (notation as in the proof of Lemma~\ref{le:Cohomology}):
\[
\xymatrix@=5mm{
&&&&0\ar[d]&\\
&&\Chains_{k+1}(X_g)/\lambda \Chains_{k}(X_g)\ar[d]^{q_{k+1}}\ar[rd]^{\partial_{k+1}}&&{\mathbf H}_{k-1}\ar[d]&\\
0\ar[r]&{\mathbf H}_k\ar[r]&\Cycles_k/\lambda\Cycles_{k-1}\ar[r]^{i_k}\ar[d]&\Chains_{k}(X_g)/\lambda \Chains_{k-1}(X_g)\ar[r]^{q_k}\ar[rd]^{\partial_k}&\Cycles_{k-1}/\lambda\Cycles_{k-2}\ar[d]^{i_{k-1}}\ar[r]&0\\
&&0&&\Chains_{k-1}(X_g)/\lambda\Chains_{k-2}(X_g)&\\
}
\]
We claim that, if $k\geq 2$, then 
the two columns and the horizontal row are exact.
The exactness of the left column follows since $q_{k+1}:\Chains_{k+1}\to{\mathbf K}_k$ is surjective. The exactness of the horizontal row and the right vertical column is easy to check by diagram chasing. For instance, to prove exactness of the horizontal row at $\Cycles_k/\lambda\Cycles_{k-1}$, if $a\in \Cycles_k$ such that $[a]\in \Cycles_k/\lambda\Cycles_{k-1}$ is in the kernel of $i_k$, then there is a $b\in \Chains_{k-1}(X_g)$ such that $a=\lambda b\in \Cycles_k$. Since $\lambda a=\lambda^2 b=0$, $a$ represents a homology class in ${\mathbf H}_k$. 

From the diagram above follows that 
\begin{align*}
H_k(U_k(Y_g))&=\ker \partial_k=q_k^{-1}({\mathbf H}_{k-1})
\cong \im(i_k)\oplus {\mathbf H}_{k-1} \\
&=
\begin{cases}
\ZZ \oplus \Cycles_{k}/\lambda \Cycles_{k-1}&\text{ if $k$ is even,}\\
(\Cycles_{k}/\lambda \Cycles_{k-1})/{\mathbf H}_{k}
\cong \lambda \Cycles_{k}&\text{ if $k$ is odd.} 
\end{cases}
\end{align*}
This proves the theorem for the case $k=i$. 

In the case $k\not=i$,
the same computation proves that, if $i<k\leq 2g$, one has
\[
H_i(U_k(Y_g))=\ker \partial_k/\im(i_k)\cong {\mathbf H}_i\cong
\begin{cases}
\ZZ &\text{ if $k$ is even,}\\
0 &\text{ if $k$ is odd.}
\end{cases}
\]
Together with Lemma~\ref{le:reduction_to_torus}, this completes the proof of the theorem. 
\end{proof}
\end{theorem}

We now wish to make a few remarks on these groups:

\begin{lemma} 
\label{remarks}
$\;$

\begin{enumerate}
\item
$\Cycles_k/\lambda \Cycles_k$ and $\lambda \Cycles_k$
are free Abelian groups. 
\item
$\Cycles_{2g}/\lambda \Cycles_{2g-1}=0\quad \text{and} \quad
\lambda\Cycles_{2g-1}=0$.
\item 
If $2\leq k\leq 2g-2$, then neither $\lambda \Cycles_k$ nor
$\Cycles_{k}/\lambda \Cycles_{k-1}$
 are finitely generated. 
 \end{enumerate}
\begin{proof}
$\,$\\
\begin{enumerate}
\item
$\lambda \Cycles_k \subset \Cycles_{k+1}\subset {\mathcal C}_{k+1}(X_{2g})$, so as a subgroup
of a free Abelian group, it is free Abelian.
The canonical short exact sequence
\[
0 \xrightarrow{} {\mathbf H}_k
 \xrightarrow{} {\Cycles}_k/\lambda\Cycles_{k-1}
 \xrightarrow{} \lambda{\Cycles}_k
 \xrightarrow{} 0
\]
splits, since $\lambda \Cycles_k$ is a free Abelian group.
It follows that $\Cycles_k/\lambda \Cycles_k$ is also a free Abelian group.

\item
Since $\Cycles_{2g}=0$, we obviously have that $\Cycles_{2g}/\lambda \Cycles_{2g-1}$
is trivial. Let $a\in \Cycles_{2g-1}$, We want to show that $\lambda a$ is trivial.
But $a=\partial b$ for some $b\in {\mathcal C}_{2g}$, so 
$\lambda a =\lambda \partial b=\partial (\lambda b)$. Since $\lambda b\in {\mathcal C}_{2g+1}=0$,
this class is zero.

\item
We need to show that if $2\leq k\leq 2g-2$, then $\lambda \Cycles_k$ is a nontrivial group.
Pick integers $1<i_1<i_2<\dots < i_m \leq g$, $1<j_1<j_2<\dots< j_n \leq g$, such that
$m+n=k$. This can be done, since $k\leq 2g-2$. Let
$a=e_1e_ {i_1}\cdots e_{i_m}f_{i_1}\cdots f_{j_n}\in {\mathcal C}_{k+1}$. The point is
that this is a monomial in $e$ and $f$ such that $e_1$ occurs, but $f_1$ does not occur.
We claim that $\lambda \partial a$ is a nontrivial element of $\lambda \Cycles_k$. To see
that it is nontrivial, consider the homomorphism
$F : {\mathcal C}_* \to \Lambda[e_1,\dots,e_g,f_1,\dots,f_g]$ given by 
$F(x_i)=1$, $i\not=1$, $F(x_1)=0$, $F(y_i)=1$, $F(e_i)=e_i$ and $F(f_i)=f_i$.
We compute
\[
F(\lambda \partial a)=F(\lambda)F(\partial a)=-f_1 \cdot e_ {i_1}\cdots e_{i_m}f_{i_1}\cdots f_{j_n}\not=0.
\]
It follows that $\lambda \Cycles_k\not=0$.
Moreover,  $\lambda \Cycles_k\subset {\mathcal C}_{n+1}$
as a submodule over the group ring $\ZZ[\FundGroup]$. There
are no zero divisors in this ring, so any nontrivial element of
${\mathcal C}_{n+1}$ generates a copy of $\ZZ[\FundGroup]$, which is not
finitely generated as an  group. It follows that   $\lambda \Cycles_k$
cannot be finitely generated.
\end{enumerate}
\end{proof}
\end{lemma}

\begin{remark} For even $k$,  
we have a short exact sequence of modules
\[
0\to \Cycles_{k}/\lambda\Cycles_{k-1}\to H_{k}(U_{k}(X_{2g}))\to \ZZ\to 0.
\]
This sequence is split as a short exact sequence of  groups,
but it is not split as a sequence of 
$\ZZ[x_{1}^{ \pm 1},\dots,x_{g}^{ \pm 1},y_{1}^{ \pm 1},\dots,y_{g}^{
  \pm 1}]$-modules. In particular, for $k=2$, Lemma 9 of \cite{Kallel}
shows that for $g=2$ we have
$H_{2}(U_{2}(X_{2g}))\cong
\ZZ[x_{1}^{ \pm 1},x_{2}^{ \pm 1},y_{1}^{ \pm 1},y_{2}^{
  \pm 1}]$, which does not contain a submodule isomorphic to $\ZZ$.
However, the result is compatible with ours. If
\[
\eta:\ZZ[x_{1}^{ \pm 1},x_{2}^{ \pm 1},y_{1}^{ \pm 1},y_{2}^{
  \pm 1}]\to \ZZ
\]
is the augmentation map, there is an isomorphism
$\alpha : \Cycles_{k}/\lambda\Cycles_{k-1} \to \ker(\eta)$.
\end{remark}

For the purposes of this paper, the main outcome of our homology computations
is the following consequence of Lemma~\ref{remarks}:

\begin{corollary} \label{CORnotfg}
Let $\Sigma$ be a closed, oriented surface of genus $g>1$. 
Assume that $2 \leq k \leq 2g-2$. Let $\widetilde{\SymmPower {k}\Sigma}$ denote the universal cover of the $k$-fold symmetric product of $\Sigma$. The homology group $H_{k}(\widetilde{\SymmPower {k}\Sigma};\ZZ)$ is a free Abelian group, but it is not finitely generated.
\end{corollary}

\section{The homotopy type of $\widetilde{{\rm Sym}^{k}(\Sigma)}$}
\label{sec:Homotopy}

In this section, we determine the homotopy type 
of the universal cover of the symmetric products of a closed oriented surface. Although not directly relevant to
the proof of our main result, we include this calculation for the sake of completeness.

\begin{theorem}
\label{th:HomotopyTypeX}
\[
\pi_i(U_k(X_n))=
\begin{cases}
0 & \text{ if $i < k$ or if $n \leq k$,}\\
\Cycles_{k}& \text{ if $i=k<n$.}\\
\end{cases}
\]
Furthermore, $U_k(X_n)$ is homotopy equivalent
to a wedge of spheres of dimension $k$.
\begin{proof}
According to Lemma~\ref{le:torus}, $U_k(X_n)$ has trivial homology in
dimensions less than $k$. Since it is also simply connected by
definition,
Hurewicz's theorem says that the Hurewicz map
\[
h:\pi_{k}(U_k(X_n))\to
H_{k}(U_k(X_n))
\]
is an isomorphism. The statements about homotopy groups 
follow from this and the 
computation in Lemma~\ref{le:torus} of the cellullar homology of
 $U_k(X)$.
To prove the last statement, chose a family of generators
$\{f_\alpha\}_{\alpha\in A}$ for 
 $\pi_k(U_k(X))$. This family defines 
a homology equivalence  $f=\vee f_\alpha : \vee S^k_\alpha\to
U_k(X)$.
Since both its source and the target are simply connected, $f$
is a homotopy equivalence.  
\end{proof}
\end{theorem}

\begin{remark}
Theorem~\ref{th:HomotopyTypeX} is compatible with
\cite{KallelSalvatore}, Lemma 3, which states that 
$\ReducedSymm kX$ is homeomorphic to the $k$-skeleton of 
the torus $\mathbb{T}^{2g}$. 
\end{remark}

\begin{lemma}
\label{le:cofibration}
Let $m:=\min\{k-g,[k/2]\}$.
There is a map $f :\vee_{i\in I}S^k \to U_k(Y_g)$ such that 
the mapping cone $C(f)$ of $f$ is homotopy equivalent to
$\CP^{m}$. If $2\leq k\leq 2g-2$, the index set $I$ is infinite.
\begin{proof}
If $k\geq 2g$, we know from a theorem by 
Mattuck~\cite{Mat}
that $U_k(Y_g)\simeq \widetilde{\SymmPower k{\Sigma_g}}$ is homotopy equivalent to
$\CP^{k-g}$. We only have to consider the case that $k<2g$. 
The stabilisation map defines a map
\[
s : U_{k}(Y_g)\to U_{2g}(Y_g) \simeq \CP^{g} \subset \CP^{\infty}.
\]
This map is an isomorphism on homology in dimensions $\leq k-1$, 
and by Lemma \ref{le:torus} and Theorem \ref{th:main} we
have an exact sequence
\[
H_k(U_{k}(X_{2g})\xrightarrow{{i_{*}}}
H_k(U_{k}(Y_g))\to 
H_{k}(\CP^{\infty})\to 0.
\]
We can chose a free  subgroup of 
$H_k(U_{k}(X_{2g}))$ that maps 
isomorphically to the image of 
$H_k(U_{k}(X_{2g}))$ 
inside $H_k(U_{k}(Y_g))$. The Hurewicz map
identifies this subgroup with a free subgroup of
$\pi_k(U_{k}(X_{2g}))$. Let $I$
be a set of generators for this subgroup. Each generator
corresponds to a homotopy class of based maps
$S^{k}\to U_{k}(X_{2g})$. Together, they define a map
$f :\vee_{i\in I}S^k \to U_k(Y_g)$ such that 
$s\circ f$ is homotopically trivial. We intend to show that 
$C(f)$ is homotopy equivalent to $\CP^{m}$.
The long exact homology sequence of the cofibration defining
$C(f)$ already shows that at least its homology is correct:
\[
H_{i}(C(f))\cong
\begin{cases}
\ZZ & \text{ if $0\leq i \leq 2m$ and $i$ is even,}\\
0&\text{ else.} 
\end{cases}
\]

The homotopy trivialisation of the composite $s\circ f$
defines a map $s^{\prime}$ from the cofibre $C(f)$ to
$\CP^{\infty}$. By obstruction theory, the obstructions
to factoring this map over $\CP^m$ lie in the groups
$H^{j}(C(f),\pi_{j-1}(S^{2m+1}))$. Since these
groups are all trivial, we can find a map 
$h:C(f)\to \CP^{m}$ such that $s^{\prime}$ is homotopic
to $i\circ h$. It follows that $h$
induces isomorphism on all homology groups. 
The assertion in the theorem now follows from this and from Whitehead's theorem. 
\end{proof}
\end{lemma}

This result is sufficient to determine the homotopy type of $U_{k}(Y_g)$.

\begin{theorem} \label{homotopytype} There are homotopy equivalences
\[
\widetilde{\SymmPower k\Sigma}\simeq U_k(Y_g)\simeq \CP^{m}\vee (\vee_{i\in I}S^k),
\]
where $m$ and $I$ are as in Lemma \ref{le:cofibration}.
\begin{proof}
According to Lemma \ref{le:cofibration} we have a cofibration sequence
up to homotopy
\[
\vee_{i\in I}S^k \xrightarrow{f}U_{k}(Y_g)
\xrightarrow{h} \CP^{m}.
\] 

We construct a splitting $s$, up to homotopy, of the map
$h:U_k(Y)\to \CP^{m}$. That is, a map $s:\CP^{m}\to U_{k}(Y_g)$ such
that the composite $s\circ f$ is homotopic to the identity.
The existence of this splitting is guaranteed by obstruction theory, since
the obstructions to the splitting are in the trivial groups
$H^i(\CP^m,\pi_{i}(C(f),U_n(\Sigma_g))$.

Consider the diagram, commutative up to homotopy
\[
\xymatrix{
\vee_{i\in I}S^{n}\ar[r]^f& U_n(Y_g)\ar[r]^h& \CP^m\\
&\CP^m\ar[u]^{s}\ar[ru]_{\simeq\mathrm{Id}}&
}
\] 
The top row is a cofibration up to homotopy. It follows that the map 
$f \vee s:\vee_{i\in I}S^{n} \vee  \CP^{m}\to U_n(Y_g)$ is a homotopy equivalence.
\end{proof}
\end{theorem}

\section{Rigidity and holomorphic bisectional curvature}
\label{sec:Rigidity}

After the topological digression of the three preceding sections,  we are equipped to give a proof of Theorem~\ref{THMholbiscur}. We start by recalling the following rigidity result of Mok~\cite{MokNHBC, MokMRT}:

\begin{theorem} \label{THMrigidity}
Suppose that $\tilde{M}$ is the universal cover of a compact K\"ahler manifold $M$ with nonnegative holomorphic bisectional curvature. Then, for suitable
$m,m_1,\ldots,m_{r}, s \in \mathbb{N}_0$, there is a  biholomorphic isometry
of $\tilde M$, equipped with the K\"ahler metric pulled back from $M$, to the
product of K\"ahler manifolds,
\begin{equation}\label{rigfactor}
(\mathbb{C}^m, e) \times \prod_{i=1}^{r} (\mathbb{CP}^{m_i},\vartheta_i) \times \prod_{j=1}^{s} (H_j,h_j) 
\end{equation}
where $e$ is the Euclidean metric,  $\vartheta_i$ are K\"ahler metrics of nonnegative holomorphic bisectional curvature, and $h_j$ are canonical metrics on compact  irreducible Hermitian symmetric spaces $H_j$ of rank $>1$. 
\end{theorem}

Recall that a Hermitian symmetric space is a Riemannian symmetric space $G/H$ for a Lie
group $G$, equipped with a parallel complex structure with respect to which the metric is Hermitian; so the complex structure is preserved by $H$ and there is an underlying K\"ahler
structure. 

The following is an easy consequence of the classification:
\begin{proposition} \label{PROPfghomology}
Let $M$ be a compact K\"ahler manifold with nonnegative holomorphic bisectional curvature.
Then the total singular homology $H_{*}(\tilde{M}; \mathbb{Z})$ of its
universal cover is finitely generated (as an  group).
\begin{proof}
According to the decomposition of Theorem \ref{THMrigidity},
the universal cover $\tilde M$ is homotopy equivalent to a product 
\begin{equation}\label{homotopyfactor}
\prod_{i=1}^{r} \mathbb{CP}^{m_i} \times \prod_{j=1}^{s} H_j 
\end{equation}
All factors are compact manifolds, so their product is also a compact manifold.
But the homology of a compact manifold is finitely generated in every
dimension. (see e.g.\ \cite{Hatcher}, Appendix, Corollary A9 and Corollary A8).

\end{proof}
\end{proposition}

\begin{proof}[Proof of Theorem~\ref{THMholbiscur}]
We fix positive integers $g$ and $k$  such that  $1<k<2g-1$, and consider the case of
the moduli space of vortices $M={\rm Sym}^k (\Sigma_g)$.  Combining Proposition~\ref{PROPfghomology} with Corollary~\ref{CORnotfg}, Theorem~\ref{THMholbiscur} follows
immediately in the situation where $k>1$. 
But as we saw in the end of the Introduction, the case $k=1$ is trivial.
\end{proof}

\section{Vortices on hyperelliptic surfaces} \label{hyperel}

In this final section, we let $\Sigma$ be a hyperelliptic Riemann surface of genus $g \geq 2$ --- for example, any smooth complex curve of genus $g= 2$. Recall that
${\rm Sym}^{k}(\Sigma)$ is a manifold such that, for $k\geq 2$, 
$\pi_{1}({\rm Sym}^{k}(\Sigma))\cong \ZZ^{2g}$ (see Lemma~\ref{THMpi1Sym}).
In \cite{KallelSalvatore}, the second homotopy group of this space is 
discussed. In particular, Lemma 9 of this paper says that
\begin{equation} \label{pi2Sym2}
\pi_{2}({\rm Sym}^{2}(\Sigma))=\ZZ[t_{i},t_{i}^{-1}],\quad 1\leq i\leq 4. 
\end{equation}
In the present work, we have generalised this formula to 
higher homotopy groups in
Theorem~\ref{th:HomotopyTypeX}.

The hyperelliptic condition can be characterised by saying that $\Sigma$ admits a
holomorphic involution
\[
\sigma: \Sigma \longrightarrow \Sigma
\]
such that $p: \Sigma\to \Sigma/\{p\sim \sigma(p)\}$ is a branched covering over 
$S^{2}\cong \Sigma/\{p\sim \sigma(p)\}$; whenever such map $\sigma$ exists, it is
unique.  Given $\sigma$ or $p$, there is an obvious holomorphic map
\begin{equation} \label{alpha}
\alpha: S^{2}\longrightarrow {\rm Sym}^{2}(\Sigma)
\end{equation}
which sends a point of the base $S^2$ of $p$ to the two (possibly equal) points 
on $\Sigma$ covering it.
The domain of this map can be interpreted as a projective line parametrising all effective divisors in
a $g^{1}_{2}$ (i.e.\ pencil of degree two, cf.~\cite{ArbCorGriHar}) associated to the hyperelliptic structure.
In Lemma 5 of
reference~\cite{Kallel}, the image $h_{*}([\alpha])$ in $H_2({\rm Sym}^{2}(\Sigma);\ZZ)$ 
of the homotopy class represented by (\ref{alpha}) under the Hurewicz
map is computed. It is further found  that this image is a non-torsion element of the second homology
group $H_2({\rm Sym}^{2}(\Sigma);\ZZ)$, so $[\alpha]$ and all its integral multiples define nontrivial elements of 
the second homotopy group
$\pi_{2}({\rm Sym}^{2}(\Sigma))$.

Notice that the hyperelliptic involution $\sigma$ induces another involution map 
\[
\sigma^{(2)}:{\rm Sym}^2(\Sigma) \longrightarrow {\rm Sym}^2(\Sigma),
\]
defined on degree two divisors on $\Sigma$ by $x+y\mapsto \sigma(x) + \sigma(y)$.
Clearly, the set of fixed points of this map $\sigma^{(2)}$ is the image of the map $\alpha$ in (\ref{alpha}).
We have the following result:

\begin{proposition}
The image of the map $\alpha$ is a noncontractible holomorphic 2-sphere
in the moduli space of two vortices on $\Sigma$. If  the map $\sigma^{(2)}$ is an isometry
of the vortex metric on ${\rm Sym}^2(\Sigma)$, then this 2-sphere is totally geodesic. 
\end{proposition}
\begin{proof}
The argument we have just given above already impies that $\alpha(S^2)$ is
not contractible in ${\rm Sym}^{2}(\Sigma)$. Since $\alpha$ is holomorphic, it embeds
the 2-sphere as a complex submanifold. It is a general result in Riemannian geometry that, if the the set of fixed points of an isometry form a submanifold, it is necessarily totally
geodesic --- this follows easily from the Picard--Lindel\"of theorem applied to geodesic
flow (see Appendix A in~\cite{RomDLC}).
\end{proof}

The induced $L^2$-metric gives a shape to this 2-sphere,
whose geometry would be interesting to understand. 
In~\cite{ManRom}, it is shown that, as $\tau \rightarrow \frac{4 \pi k}{{\rm Vol}(\Sigma_g)}$
(the ``dissolving limit''), the metric degenerates: null-vectors appear precisely along the tangent directions to
this 2-sphere as the limit is attained.
A number of natural questions
can be asked, for example: which metrics on $\Sigma$ will ensure that this sphere
is round, or that it admits nontrivial isometries?

\vspace{18pt}

\noindent
{\sc Acknowledgements:} NMR would like to thank Jo\~ao Baptista for discussions in the early stages of this project, Indranil Biswas for clarifying some questions, and Martin Speight for a helpful comment. This work was partially supported by CTQM, University of Aarhus, and by the European Commission in the framework of the Marie Curie project MTKD-CT-2006-042360.

\bibliographystyle{numsty}

\end{document}